\documentclass[conference]{IEEEtran}
\IEEEoverridecommandlockouts

\usepackage{cite}
\usepackage{amsmath}
\usepackage{amsthm}
\usepackage{amssymb}

\DeclareMathOperator*{\argmin}{arg\,min}
\usepackage{graphicx}
\usepackage{textcomp}
\usepackage{xcolor}

\usepackage{algorithm}
\usepackage[noend]{algpseudocode}
\makeatletter
\def\BState{\State\hskip-\ALG@thistlm}
\makeatother
\newtheorem{proposition}{Proposition}
\newtheorem{definition}{Definition}
\newtheorem{remark}{Remark}

\usepackage{matlab-prettifier}
\usepackage{color}
\usepackage{mathrsfs}
\usepackage[shortlabels]{enumitem}
\makeatletter
\def\BState{\State\hskip-\ALG@thistlm}
\makeatother
\def\BibTeX{{\rm B\kern-.05em{\sc i\kern-.025em b}\kern-.08em
    T\kern-.1667em\lower.7ex\hbox{E}\kern-.125emX}}
\usepackage{subcaption}

\begin{document}

\title{{\bf\Large Strategic Information Attacks on Incentive-Compatible Navigational Recommendations in Intelligent Transportation Systems}
}

\author{Ya-Ting Yang, Haozhe~Lei, and Quanyan Zhu
\thanks{The Authors are with the Department of Electrical and Computer Engineering, New York University, Brooklyn, NY, 11201, USA; E-mail: {\tt\small \{yy4348, hl4155, qz494\}@nyu.edu}. YY and HL have contributed equally. Correspondence should be addressed to YY.}%
}

\maketitle

\begin{abstract}
Intelligent transportation systems (ITS) have gained significant attention from various communities, driven by rapid advancements in informational technology. Within the realm of ITS, navigational recommendation systems (RS) play a pivotal role, as users often face diverse path (route) options in such complex urban environments. However, RS is not immune to vulnerabilities, especially when confronted with potential information-based attacks. This study aims to explore the impacts of these cyber threats on RS, explicitly focusing on local targeted information attacks in which the attacker favors certain groups or businesses. We study human behaviors and propose the coordinated incentive-compatible RS that guides users toward a mixed Nash equilibrium, under which each user has no incentive to deviate from the recommendation. Then, we delve into the vulnerabilities within the recommendation process, focusing on scenarios involving misinformed demands. In such cases, the attacker can fabricate fake users to mislead the RS's recommendations. Using the Stackelberg game approach, the analytical results and the numerical case study reveal that RS is susceptible to informational attacks. This study highlights the need to consider informational attacks for a more resilient and effective navigational recommendation.
\end{abstract}


\section{Introduction}
Harnessing the vast information available from modern wireless communication and Internet of Things (IoT) advancements \cite{lv2020ai,zantalis2019review}, coupled with the progress made in data science and artificial intelligence \cite{veres2019deep,haydari2020deep}, intelligent transportation systems (ITS) have gained substantial attention for their ability to effectively tackle traffic congestion and elevate driver experiences. Within the realm of ITS, path recommendation systems (RS) such as Google Maps and Apple Maps play a vital role in complex urban environments with diverse route choices \cite{van2016user} for users, including drivers and pedestrians. Based on the given information, the RS offers routes to simplify users' decision-making processes by presenting good paths, aiming to reduce travel duration and alleviate congestion \cite{8262884}. However, unlike routing, user compliance with recommendations is not guaranteed, emphasizing the need for incentive-compatible recommendations \cite{ning2023robust}.

In addition, the path recommendation process is prone to various vulnerabilities \cite{mecheva2020cybersecurity} that attackers can leverage to promote particular groups or businesses in a local targeted sense or potentially exacerbating congestion levels on a broader network-wide scale. Within this context, information-based attacks emerge as a critical concern, as they empower malicious entities to spread misinformation and manipulate drivers to achieve their objectives \cite{Waniek2021}. For example, within the Waze platform, police officers can intentionally provide no congestion information on certain roads to catch speeding drivers \cite{waze_police}, while residents may fabricate congestion reports to divert traffic away from their residential areas \cite{waze_resident}. Moreover, \cite{eryonucu2022sybil} illustrates how Sybil-based attacks can effectively manipulate the crowdedness at the points of interest, traffic congestion levels, and public transportation occupancy within Google Maps.

In this work, we discuss the potential vulnerabilities of the coordinated incentive-compatible navigational RS. Our initial step involves formulating a comprehensive RS problem. Beyond the traditional RS approaches that typically optimize individual user costs or the overall social profile \cite{Multihoming,zhang2010optimizing}, our approach takes into account the impact of human behaviors. Our primary goal is to ensure that users do not have incentives to deviate from the recommendation, while considering the non-user drivers. From the established framework, we observe an equivalence between the incentive-compatible RS and the Wardrop equilibrium \cite{WE}, a connection that facilitates the risk analysis and the design of computational algorithms.

To assess the risk of cyber attacks on RS, particularly in the context of local information attacks, we employ Stackelberg game frameworks. The attacker is the leader who conveys misinformation, while the RS, as the follower, responds to the provided information. Although attacks may differ in their methods, their impacts can be classified into two major categories - misinformed demands and misinformed traffic conditions such as travel time costs. This study focuses on the former, in which the attacker can fabricate Sybil-based users with fake origin-destination (OD) pairs to make the RS recommend the true users towards a specific road that favors certain groups or businesses. That is, by carefully designing the misinformed demands, such as how many fake users with which OD pairs, the attacker can make the RS redistribute the true users on other alternative paths towards the target road.


This research demonstrates that the RS considering human behaviors and recommending users toward a mixed Nash equilibrium is susceptible to informational attacks. Our contributions can be summarized as follows.
\begin{itemize}
    \item We formulate a game-theoretic problem aimed at designing an incentive-compatible navigational RS that considers the behaviors of non-users within the traffic network and guides users toward an optimal traffic equilibrium where no users have incentives to deviate from the provided recommendations.
    \item We identify vulnerabilities of the RS, with a focus on the misinformed demands in the local targeted attack. To quantify the risk of these vulnerabilities, we employ a Stackelberg game approach. Both analytical results and numerical case studies show that users are highly vulnerable to attacks targeting roads by fabricating fake demands with OD pairs on alternative paths.
    \item We discover a ``Resilience Paradox'' where the local targeted attack by misinformed demands or traffic conditions can benefit the overall traffic outcome regarding total travel time costs in some situations.

 
\end{itemize}

\section{Related Works}
Research efforts, e.g., \cite{van2016user,siuhi2016opportunities}, on navigational RSs share the same goal of elevating user experience and mitigating the congestion level. One aspect typically focuses on optimizing recommendations for independent routing, which tends to overlook other users and may lead to the flash crowd effect \cite{flash_crowd_effect}. Alternatively, the RS may prioritize system efficiency at the expense of some users' utilities, or may consider a user-oriented equilibrium routing \cite{user_eqm_rl, ning2023robust} that reduces system costs to a certain level compared with independent routing. Our study aims to go one step further. We consider human factors such as user's compliance with recommendations and the behaviors of those non-user drivers, then propose a more comprehensive RS that guides users towards a mixed Nash equilibrium.


Regarding malicious entities or potential attackers \cite{mecheva2020cybersecurity} in the realm of ITS \cite{cybersecurity_transit}, most studies typically have focused on attacks that aim to increase the overall congestion level of the traffic network \cite{pan2022poisoned}. We aim to scrutinize the particular vulnerabilities inherent to navigational RSs, which are susceptible to a wide range of potential attacks, and focus on the local targeted attacks that has few systematic studies. Specifically, this work expands the scope of recent studies \cite{eryonucu2022sybil} and identifies a broadened class of attacks, including attackers who may seek to mislead users onto specific roads. We delve into Sybil-based attack methods, where misinformation regarding non-existent demands comes from fake users.





\section{System Model and Preliminary}\label{sec:model}

The feedback structure of the coordinated incentive-compatible navigational RS is illustrated in Fig. \ref{fig:RS_model}. When seeking path recommendations, users begin by providing their OD pairs to the RS. The RS collects data from various sources, including traffic sensors, cloud databases, and user reports. Then, the RS processes this information to generate path recommendations for the users.

\subsection{Navigational Recommendation Systems} \label{sec:NRS}
Motivated by the congestion game \cite{ning2023robust}, the ingredients for a RS consist of the following. The RS is the app that recommends paths to its ``users'' $u \in \mathcal{U}=\{1, \cdots, m\}$. Each user $u$ has an origin $O_{u}$ and destination $D_{u}$ pair (OD pair), denoting the OD pair of user $u$ as $\theta_u$ for later use. Each user $u$ then has a set of feasible paths $\mathcal{S}_u=\{1, \cdots, k_u\}$ from origin $O_{u}$ and destination $D_{u}$. Such a set can be identified by shortest path algorithms \cite{yen1971finding}. The urban transportation network can be represented as a graph $\mathcal{G}=(\mathcal{V}, \mathcal{E})$, where $\mathcal{V}$ denotes the set of intersections; $\mathcal{E}$ represents roads between intersections. Passing through an edge $e\in \mathcal{E}$ induces a cost $c_e: \mathbb{R}_{\ge 0} \mapsto \mathbb{R}_{+}$ related to the expected flow $f_e \in \mathbb{R}_{\ge 0}$ on that road $e$. One possible choice of the cost function $c_e(\cdot)$ can be the travel time cost $c_e(f_e)=t_e\left(1+\eta\left(\frac{f_e}{k_e}\right)^\zeta\right)$ that given by the standard
Bureau of Public Roads (BPR) function, where $t_e \in \mathbb{R}_{+}$ is the free-flow travel time on edge $e$, $k_e \in \mathbb{R}_{+}$ is the capacity of edge $e$, and $\eta, \zeta \in \mathbb{R}_{\ge 0}$ are some parameters.

\begin{figure}
    \centering
    \includegraphics[width=3.1in]{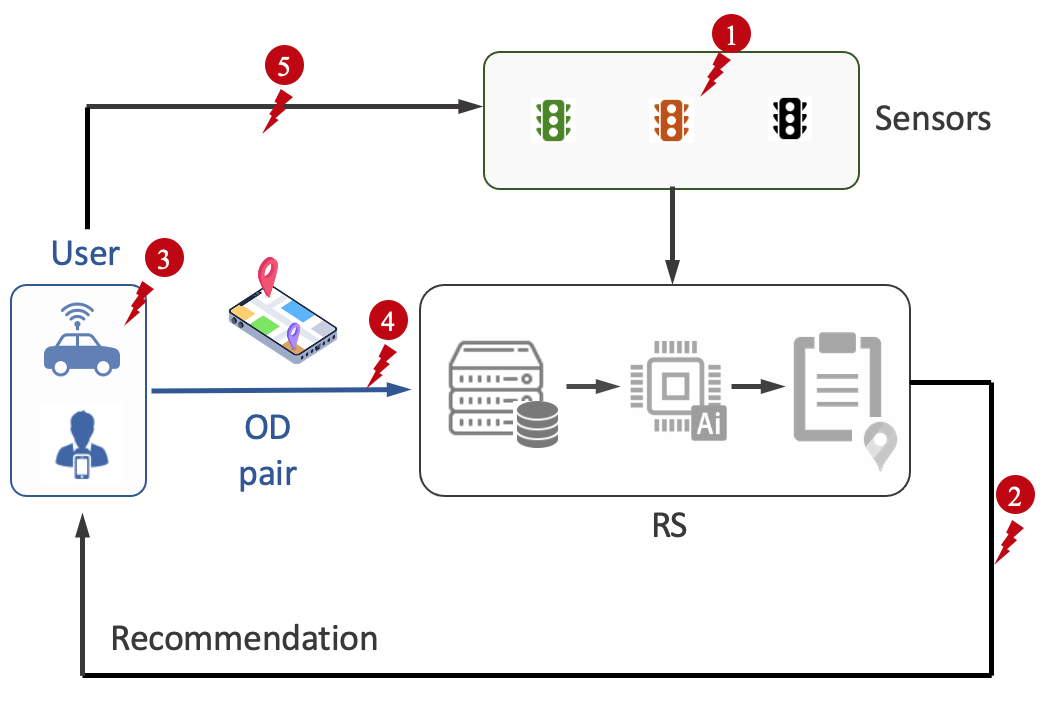}
    \caption{The process for path recommendations. The users report their OD pairs to the RS. The RS gathers data (from sensors, cloud database, other users’ report, etc.) and then provide path recommendations  to the users; (1)-(5) indicate possible vulnerabilities of the RS.}
    \label{fig:RS_model}
\vspace{-5mm}
\end{figure}

In summary, the structural objects that the RS considers is given by $\mathscr{G}=\left\langle \ \mathcal{U}, (\mathcal{S}_u)_{u \in \mathcal{U}}, \mathcal{G}, (c_e(\cdot))_{e \in \mathcal{E}} \ \right\rangle$. In this study, we use a mixed strategy for the stochastic choice behavior of the app users. Denote user $u$'s strategy as a probability mass function over the feasible paths set $\textbf{P}_u \in \Delta (\mathcal{S}_u)$. Here, $\textbf{P}_u=\{p_{u, i}\}_{i=1, \cdots, k_u}$, and $\mathcal{P}_u:=\Delta (\mathcal{S}_u)$ with $\mathcal{P}:=\Pi_{u \in \mathcal{U}}\mathcal{P}_u$. That is, $p_{u, i}\in[0,1]$ represents the probability that user $u$ places on the path $s_{u, i} \in \mathcal{S}_u$ with the constraints
\begin{equation}
    \sum_{i=1}^{k_u}p_{u, i}=1 , \ \forall u \in \mathcal{U}.
\end{equation}

Note that the probability $p_{u, i}$ can also be viewed as the expected volume generated by user $u$ on path $s_{u, i}$, which leads to the expected flow on edge $e$ as 
\begin{equation}
    f_e=\sum_{u \in \mathcal{U}}\sum_{i=1}^{k_u}p_{u, i}\mathbf{1}_{\{e \in s_{u, i}\}}.
\end{equation} Therefore, a generalized travel cost $C_{u, i}: \mathcal{P} \mapsto \mathbb{R}_{+}$ for user $u$ with each path $s_{u, i}$ can be written as the one that sums over all edge costs of the path:
\begin{equation}
    C_{u, i}(\textbf{P})=\sum_{e \in s_{u, i}}c_e(f_e),
\end{equation} for all app users' path choice probabilities $\textbf{P} \in \mathcal{P}$, with $\textbf{P}=\{\textbf{P}_u\}_{u \in \mathcal{U}}$.

In addition to those app users $\mathcal{U}$, the traffic condition, or the expected flow on each road, is influenced by other ``drivers'' in the urban transportation network. Hence, we denote the set of other drivers who dispense with the RS as $\Bar{\mathcal{U}}$. Similarly, each driver $\Bar{u} \in \Bar{\mathcal{U}}$ has an origin $O_{\Bar{u}}$ and destination $D_{\Bar{u}}$ pair that determines a set of feasible paths $S_{\Bar{u}} = \{1,\cdots , k_{\Bar{u}}\}$. Without loss of generality, we assume that those drivers' stochastic choice behavior can be modeled by the multinomial logit (\textbf{MNL}) model \cite{daskin1985urban}, which has been widely used in various fields such as economics, transportation \cite{lee2018comparison}, marketing, and social sciences, to analyze discrete choice among multiple alternatives. In the context of path choices, the model supposes that individuals make choices based on the cost they associate with each available option. Hence, based on the current (initial) cost (travel time) $C_{\Bar{u}, i}^o \in \mathbb R_{+}$ on path $s_{\Bar{u}, i}$, driver $\Bar{u}$'s path choice preference without recommendations is 
\begin{equation}
    p_{\Bar{u}, i}^o=\frac{e^{V_{\Bar{u}, i}}}{\sum_{i=1}^{k_{\Bar{u}}}e^{V_{\Bar{u}, i}}},
\label{eq:preference}
\end{equation} Note that the term $V_{\Bar{u}, i}=-\alpha_{\Bar{u}}-\beta_{\Bar{u}} C_{{\Bar{u}}, i}^o$, where $\alpha_{\Bar{u}} \in \mathbb{R}$ and $\beta_{\Bar{u}} \in \mathbb{R}$ indicate driver $\Bar{u}$'s valuations. Then, we denote $\textbf{P}^{o}=\{p_{\Bar{u}, i}^o\}_{ \Bar{u} \in \Bar{\mathcal{U}}, i=1, \cdots, k_{\Bar{u}}}$ as all drivers' path choice preferences.

Therefore, the recommendation $\textbf{P}^r=\{p_{u, i}^r\}_{ u \in \mathcal{U}, i=1, \cdots, k_{u}}$ to all the app users can be a \textbf{feasible} one if it satisfies the following constraints:
\begin{subequations}
  \begin{align}
    & \sum^{k_u}_{i=1}p^r_{u,i}\left(C_{u,i}(\textbf{P}^r_{u},\textbf{P}^r_{-u}, \textbf{P}^o)\right)-p^d_{u,i}\left(C_{u,i}(\textbf{P}^d_{u},\textbf{P}^r_{-u}, \textbf{P}^o)\right)\nonumber\\ 
    &\qquad \leq 0, 
      \forall u \in \mathcal{U}, \forall \ \textbf{P}^d_{u}=\{p^d_{u,i}\}_{i \in \{1,\cdots,k_u\}} \in \mathcal{P}_u, \label{eq:correl_cons}\\
    &\sum^{k_u}_{i=1}p^r_{u,i}=1, \forall u \in \mathcal{U}, \label{eq:prob_sum}\\
    & p^r_{u,i}\geq 0, \forall u \in \mathcal{U}, \forall s_{u, i} \in \mathcal{S}_u,\label{eq:prob_geq0}
  \end{align}
\label{prob:RS}
\end{subequations}with 
\begin{equation}
    f^r_e=\sum_{u \in \mathcal{U}}\sum_{i=1}^{k_u}p^r_{u, i}\mathbf{1}_{\{e \in s_{u, i}\}}, \forall e \in \mathcal{E}, \label{eq:flow_r}
\end{equation} denotes the expected flow caused by the recommendations to users on each road, and
\begin{equation}
    f^o_e=\sum_{\Bar{u} \in \Bar{\mathcal{U}}}\sum_{i=1}^{k_{\Bar{u}}}p^o_{\Bar{u}, i}\mathbf{1}_{\{e \in s_{\Bar{u}, i}\}}, \forall e \in \mathcal{E}, \label{eq:flow_o}
\end{equation} indicates the expected flow caused by other drivers on the road. The expected flow caused by users and drivers result in the cost of each user $u$'s choice $s_{u, i}$:
\begin{equation}
    C_{u, i}(\textbf{P}^r)=\sum_{e \in s_{u, i}}c_e(f^r_e+f^o_e), \forall u \in \mathcal{U}, \forall s_{u, i} \in \mathcal{S}_u, \label{eq:path_cost}
\end{equation} where the recommendation to user $u$ is $\textbf{P}^r_{u} \in \mathcal{P}_u$, and the recommendations to other users except $u$ is $\textbf{P}^r_{-u} \in \Pi_{u^{\prime} \in \mathcal{U}\setminus\{u\}}\mathcal{P}_{u^{\prime}}$. It is important to note that constraint \eqref{eq:correl_cons} coincides with the definition of Nash equilibrium, where each user will have no incentive to unilaterally deviate from the recommended mixed strategy. This RS takes into account the preferences of the users as a group and creates coordinated incentive-compatible recommendations, which differs from the recommendation of the shortest path to all users. The constraints \eqref{eq:prob_sum} and \eqref{eq:prob_geq0} ensure that $\{p_{u, i}^r\}_{ u \in \mathcal{U}, i=1, \cdots, k_{u}}$ are valid mixed strategies.

\subsection{Feasibility Analysis}


To analyze the feasibility of problem \eqref{prob:RS}, we can connect it with the Wardrop equilibrium \cite{WE}. 

\begin{definition}[Wardrop equilibrium]
    A Wardrop equilibrium (WE) is a feasible path flow and road load pair $(\boldsymbol{y}, \boldsymbol{f})$ with vectors $\boldsymbol{y}\in\mathbb{R}_{\ge 0}^{\Pi_{t \in \mathcal{T}}|\mathcal{S}_t| }$ and $\boldsymbol{f}\in \mathbb{R}_{\ge 0}^{|\mathcal{E}|}$, where $\boldsymbol{y}:=\{y_{t, i}\}_{t \in \mathcal{T}, s_{t, i} \in \mathcal{S}_t}$ and $\boldsymbol{f}:=\{f_e\}_{e \in \mathcal{E}}$ for which the prevailing cost of all used strategies is minimal, or, mathematically, for  demand $t$ and corresponding strategy set $\mathcal{S}_t$,
\begin{equation*}
    \begin{aligned}
        & \forall t \in \mathcal{T}, \forall s_{t, i}, s_{t, j} \in \mathcal{S}_t, \quad y_{t, i}>0 \\
        & \Longrightarrow \sum_{e \in s_{t, i}} c_e\left(f_e\right) \leq \sum_{e \in s_{t, j}} c_e\left(f_e\right),
    \end{aligned}
\end{equation*}
where the path flow vector $\boldsymbol{y}$ and the road load vector $\boldsymbol{f}$ need to satisfy the following constraints:
\begin{equation}
    \begin{aligned}
        d_t &=\sum_{s_{t, i} \in \mathcal{S}_t} y_{t, i} \text { with } y_{t, i} \geq 0, \quad \forall t \in \mathcal{T}, \\
        f_e & =\sum_{t \in \mathcal{T}} \sum_{s_{t, i} \in \mathcal{S}_t} y_{t, i} \mathbf{1}_{\{e \in s_{t, i}\}}, \quad \forall e \in \mathcal{E}.
    \end{aligned}
\label{eq:F_d}
\end{equation}

\label{def:WE}
\end{definition}

\begin{proposition}
    A Wardrop equilibrium flow-load pair $(\boldsymbol{y}, \boldsymbol{f})$ in Definition \ref{def:WE} that corresponds to the recommendation-load $(\textbf{P}^r$, $(f^r_e)_{e \in \mathcal{E}})$ is a feasible solution for the incentive-compatible navigational RS described in \eqref{prob:RS}-\eqref{eq:path_cost}.
\end{proposition}
\begin{proof}

A correspondence can be observed between Definition \ref{def:WE} and the problem with the RS in \eqref{prob:RS}. That is, type $t$ and $\mathcal{S}_t$ in the definition correspond to user $u$ and the set of feasible paths $\mathcal{S}_u$, the road load vector $f_e$ in the definition is the same as the expected flow $f_e$ on edge $e$, and the path flow $y_{t, i}$ in the definition corresponds to the probability $p_{u, i}$ that user $u$ places on the path $s_{u, i}$. Hence, letting $d_t=1$ in the definition will lead to $\sum^{k_u}_{i=1}p^r_{u,i}=1$ that naturally satisfies constraint \eqref{eq:prob_sum}.

As a result, the WE flow-load pair $(\boldsymbol{y}, \boldsymbol{f})$ that corresponds to the ($\textbf{P}^r$, $(f^r_e)_{e \in \mathcal{E}}$) pair in the problem can be viewed as a feasible solution for the RS. More specifically, if ($\textbf{P}^r$, $(f^r_e)_{e \in \mathcal{E}}$) is a WE, then for each user $u$ in constraint \eqref{eq:correl_cons}:
Since only paths with minimum cost are used, all the paths used by any given user have the same cost. That is, for $p_{u, i} > 0$, the cost $C_{u,i}(\textbf{P}^r_{u},\textbf{P}^r_{-u}, \textbf{P}^o)$ should be the same for $u$. The overall expected cost is independent of the probability $p_{u, i}$ of $C_{u,i}(\textbf{P}^r_{u},\textbf{P}^r_{-u}, \textbf{P}^o)$. Lastly, note that if ($\textbf{P}^r$, $(f^r_e)_{e \in \mathcal{E}}$) is an equilibrium, there is no incentive for a user to deviate to any other $\textbf{P}^d_u \in \mathcal{P}_u$.
\end{proof}


Note that the connection can also be observed through variational inequalities \cite{sorin2015finite}. To begin with, denote the expected cost $\sum^{k_u}_{i=1}p^r_{u,i}\left(C_{u,i}(\textbf{P}^r_{u},\textbf{P}^r_{-u}, \textbf{P}^o)\right)$ in problem \eqref{prob:RS} as $H_u(\textbf{P}^r_{u}, \textbf{P}^r_{-u})$ for each user $u$, and an equilibrium is the recommendation $\textbf{P}^r=\{p_{u, i}^r\}_{ u \in \mathcal{U}, i=1, \cdots, k_{u}}$ that satisfies $H_u(\textbf{P}^r_{u}, \textbf{P}^r_{-u}) \leq H_u(\textbf{P}^d_{u}, \textbf{P}^r_{-u}), \forall \ \textbf{P}^d_{u} \in \mathcal{P}_u, \forall u \in \mathcal{U}$. Let $\nabla_u H_u(\textbf{P}^r)$ represent the gradient of $H_u(\textbf{P}^r_{u}, \textbf{P}^r_{-u})$ with respect to each element of $\textbf{P}^r_{u}$ and $\frac{\partial H_u(\textbf{P}^r)}{\partial p^r_{u, i}}=C_{u, i}(\textbf{P}^r)+\sum_{j \in \mathcal{S}_u}p^r_{u, j}\frac{\partial C_{u, j}(\textbf{P}^r)}{\partial p^r_{u, i}}$. Then, a solution of \eqref{eq:correl_cons} satisfies 
\begin{equation}
    \sum_{u \in \mathcal{U}} \langle \nabla_u H_u(\textbf{P}^r), \textbf{P}^r_u -\textbf{P}^d_u \rangle \leq 0, \forall \ \textbf{P}^d \in \mathcal{P}.
\end{equation} Similarly, from the definition of WE, let each type $t$ be associated with one corresponding user $u$. Then, we can get 
\begin{equation}
    \sum_{u \in \mathcal{U}} \langle C_u(\textbf{P}^r), \textbf{P}^r_{u}-\textbf{P}^d_{u} \rangle \leq 0, \forall \ \textbf{P}^d \in \mathcal{P},
\end{equation} where $C_u(\textbf{P}^r)=\{C_{u, i}(\textbf{P}^r)\}_{s_{u, i} \in \mathcal{S}_u}$. Then, the equilibrium for the RS's problem and the WE can be treated as equivalent in the following.
\begin{definition}
For the game $\Gamma(\mathscr{G}, \Phi)$ with structural objects $\mathscr{G}$ in Section \ref{sec:NRS} and evaluation functions $\Phi=\{\Phi_u\}_{u \in \mathcal{U}}$, where $\Phi_u: \mathcal{P} \mapsto \mathbb{R}_{+}^{|\mathcal{S}_u|}, \forall u \in \mathcal{U}$, the Nash equilibria $NE(\Phi)$ is the set of $\textbf{P}^r \in \mathcal{P}$ satisfying $\sum_{u \in \mathcal{U}}\langle \Phi_u(\textbf{P}^r), \textbf{P}^r_u-\textbf{P}^d_u \rangle \leq 0, \forall \  \textbf{P}^d \in \mathcal{P}$.
\label{def:NE_Phi}
\end{definition} The evaluation functions are $\{\nabla_u H_u\}_{u \in \mathcal{U}}$, $\nabla_u H_u: \mathcal{P} \mapsto \mathbb{R}_{+}^{|\mathcal{S}_u|}$ for the RS and $\{C_u\}_{u \in \mathcal{U}}$, $C_u: \mathcal{P} \mapsto \mathbb{R}_{+}^{|\mathcal{S}_u|}$ in WE.

\subsection{Solving for the Equilibrium}

Knowing that the feasible recommendation for problem \eqref{prob:RS} exists, we then proceed to find the mixed strategy $\textbf{P}^r_{u}$ for each user $u$. First, for user $1$'s constraint \eqref{eq:correl_cons} in the RS's problem:
$$\sum^{k_1}_{i=1}p^r_{1,i}\left(C_{1,i}(\textbf{P}^r_{1},\textbf{P}^r_{-1}, \textbf{P}^o)\right)-p^d_{1,i}\left(C_{1,i}(\textbf{P}^d_{1},\textbf{P}^r_{-1}, \textbf{P}^o)\right) \leq 0.$$ 
Given other users' strategy $\textbf{P}^r_{-1}$, user $1$'s best response (or the best recommendation to user $1$) is the following.
\begin{equation}
\begin{aligned}
    \argmin_{p^r_{1, 1}, \cdots, p^r_{1, k_1}} &\sum^{k_1}_{i=1}p^r_{1,i}C_{1,i}(\textbf{P}^r_{1},\textbf{P}^r_{-1}, \textbf{P}^o),\\
    \text{s.t. } &\sum^{k_1}_{i=1}p^r_{1,i}=1,\ \ p^r_{1,i}\geq 0, \forall s_{1, i} \in \mathcal{S}_1.
\end{aligned}
\label{eq:user1} 
\end{equation} 
Here, the problem in \eqref{eq:user1} computes mixed strategies $p^r_{1, i}, \cdots, p^r_{1, k_1}$ that minimize user $1$'s expected cost. Then, other users' constraints \eqref{eq:correl_cons} in the RS's problem follow a similar procedure. Note that the additional expected flow resulting from other drivers' path choice preferences $\textbf{P}^o$ does not change with $\textbf{P}^r$. More specifically, $c_e(f_e)=t_e(1+\eta(\frac{f_e}{k_e})^{\zeta})$ can be rewritten as $c_e(f^r_e, f^o_e)=t_e(1+\eta(\frac{f^r_e+f^o_e}{k_e})^{\zeta})$, where $f^r_e$ comes from users' mixed strategies $\textbf{P}^r$, and $f^o_e$ remains a constant comes from other drivers' choice preferences $\textbf{P}^o$. 

To this end, since \eqref{eq:user1} is a constrained optimization problem, we can use the penalty methods (PM) or other techniques to solve the problem for each user, and iteratively find the equilibrium solution. The evolution of the iterative process adopts the best response dynamics 
$\mathcal{D}^{\Phi}=\{\mathcal{D}^{\Phi}_u\}_{u \in \mathcal{U}}$, where at iteration $n$ and $n+1$ for user $u$,
\begin{equation}
\mathcal{D}^{\Phi}_u(\textbf{P}^{r, (n)})=BR_u(\textbf{P}^{r, (n)})-\textbf{P}^{r, (n)}_u,
\label{eq:BR_dynamics}
\end{equation} with 
$BR_u(\textbf{P}^{r, (n)})\! =\! \{\textbf{P}^{r,(n+1)}_u\! \in\! \mathcal{P}_u, \langle \textbf{P}^{r,(n+1)}_u-\textbf{P}^d_u, \Phi(\textbf{P}^{r, (n)}) \rangle $ $ \leq 0, \forall \textbf{P}^d_u \in \mathcal{P}_u \}$, and can be shown to satisfy Nash stationarity.
\begin{definition}[Nash Stationarity]
    For the game $\Gamma(\mathscr{G}, \Phi)$ with evaluation functions $\Phi=\{\Phi_u\}_{u \in \mathcal{U}}$, where $\Phi_u: \mathcal{P} \mapsto \mathbb{R}_{+}^{|\mathcal{S}_u|}$, the dynamics $\mathcal{D}^{\Phi}=\{\mathcal{D}^{\Phi}_u\}_{u \in \mathcal{U}}$ with $\mathcal{D}^{\Phi}_u: \mathcal{P} \mapsto [-1, 1]^{|\mathcal{S}_u|}$ satisfy Nash stationarity if $\mathcal{D}^{\Phi}(\textbf{P}^r)=0$ if and only if $\textbf{P}^r \in \mathcal{P}$ is an equilibrium $NE(\Phi)$.
\label{def:Nash_stationary}
\end{definition} 

\begin{proposition}
    The RS's iterative process that adopts the best response dynamics $\mathcal{D}^{\Phi}=\{\mathcal{D}^{\Phi}_u\}_{u \in \mathcal{U}}$ in \eqref{eq:BR_dynamics} satisfies Nash stationarity as in Definition \ref{def:Nash_stationary}.
\end{proposition}
\begin{proof}
    By Definitions \ref{def:NE_Phi} and \ref{def:Nash_stationary}, since $\mathcal{D}^{\Phi}(\textbf{P}^{r})=0$ in \eqref{eq:BR_dynamics} if and only if $\textbf{P}^{r}_u \in BR_u(\textbf{P}^{r}), \forall \ u \in \mathcal{U}$, $\textbf{P}^r\in NE(\Phi)$.
\end{proof}

The procedure for solving the RS's problem is described in Algorithm 1, and the numerical results are shown and discussed in Section \ref{sec:experiment}.


\begin{algorithm}
  \caption{\textbf{R}ecommendation \textbf{S}ystem}\label{algo:RS}
  \begin{algorithmic}[1]
    \State\textbf{Input} $\left\langle\mathcal{U}, (\mathcal{S}_u)_{u \in \mathcal{U}}, \mathcal{\Bar{U}}, (\mathcal{S}_{\Bar{u}})_{\Bar{u} \in \mathcal{\Bar{U}}}, \mathcal{G}, (c_e(\cdot))_{e \in \mathcal{E}}\right\rangle$,
    \State\textbf{Initialize} recommendation $\textbf{P}^r$,
      \State\textbf{Obtain} $\textbf{P}^o$ for other drivers $\mathcal{\Bar{U}}$,
      \While{$n < \text{MaxIter}$}
        \State choose $u=n\mod|\mathcal{U}|$,
        \State solve $\argmin_{\textbf{P}^{r}_u} \sum^{n_u}_{i=1}p^{r}_{u,i}C_{u,i}(\textbf{P}^{r}_{u},\textbf{P}^{r,(n)}_{-u}, \textbf{P}^o)$,
        \State update $\textbf{P}^{r,(n+1)}=\textbf{P}^{r,(n+1)}_{u}\cup \textbf{P}^{r,(n)}_{-u}$
      \EndWhile
      \State \textbf{Return} optimal $\textbf{P}^r$
  \end{algorithmic}
\end{algorithm}
\vspace{-3mm}

\section{Attack Models}
\subsection{Attack Methods and Consequences}
As the RS offers recommendations based on users' OD pairs, transportation network, and traffic conditions, it becomes susceptible to a range of vulnerabilities, as depicted in red in Fig.  \ref{fig:RS_model}. Those potential attacks can be but not limited to the following: 1) Sensor Manipulation Attacks: These involve tampering with traffic sensors to manipulate the collected data, consequently impacting the \textit{cost} used by the RS; 2) Communication Attacks (e.g., Man-in-the-Middle): These types hijack and alter the path recommendations for users; 3) Misinformation Attacks on Drivers: These attacks induce changes in drivers' choice preferences through misinformation (e.g., fabricated accident reports or anonymous biased reviews), leading to shifts in \textit{expected flow and cost} as evaluated by the RS; 4) Misinformation Attacks on Demand (e.g., Sybil-based attack): Attackers can fabricate non-existent users with fictitious demands; 5) Feedback Availability: Attackers could exploit delays in information structure or initiate denial-of-service attacks on the sensor data.

Although \textit{these attacks may differ in their methods, their consequences can be classified into two major categories: misinformed costs and misinformed demands}. In this study, we focus on the misinformed demands. The attacker is capable of creating Sybil users with certain origins and destinations that change the demand of the corresponding OD pairs, which will influence the resulting recommendation, that is, the feasible ($\textbf{P}^r$, $(f^r_e)_{e \in \mathcal{E}}$) pair in the RS's problem will be different.

\subsection{Attacker's Objectives}
Imagine an attacker driven by self-interest, in conflict with the overall social welfare goal of reducing congestion. This scenario can be studied at both local targeted and network-wide levels: the former pertains to specific groups or locations, while the latter considers the system-wide impact.

\noindent\textbf{Targeted Attacks: }The attacker seeks to bias the system by suggesting paths that favor particular groups (e.g., higher-paying users) or businesses (e.g., those paying the attacker to ensure users see particular ads or pass by their shops). This manipulation promotes or harms specific interests.

\noindent\textbf{Network-Wide Attacks: }The attacker  aims to disrupt the system by increasing delays or congestion indices across the network, consequently raising the overall traffic time cost. These actions could harm the system's reputation, leading to user dissatisfaction or a loss of trust in the RS.

In our ensuing discussion, we will delve into the local perspective, as existing literature such as \cite{pan2022poisoned} has already focused on the network-wide attack.

\subsection{Local Attack through Misinformed Demands}\label{sec:md_problem}

For the sake of simplicity, we suppose that the RS selects the WE flow-load pair as a practical recommendation, with the effects introduced by other non-user drivers already encompassed within the cost function for the ensuing discussion. Then, the RS's problem can be reformulated as the following using the notations from Definition \ref{def:WE}.

\subsubsection{Attacking the WE-based RS}
Let a set $\mathcal{T} \subset \mathcal{V} \times \mathcal{V}$ represents the origin-destination (OD) pairs. For each OD pair $t \in \mathcal{T}$, a flow of demand from user $d_t=\sum_{u \in \mathcal{U}} \mathbf{1}_{\{\theta_u=t\}}$ must be routed from the corresponding origin to its destination. Then, the feasible path set for each $t$ is $\mathcal{S}_t=\{1, \cdots, k_t\}$. We denote $F(\boldsymbol{d})$ as the set of flow-load pair $(\boldsymbol{y}, \boldsymbol{f})$ that satisfies both constraints in \eqref{eq:F_d} with the demand vector $\boldsymbol{d}:=\left(d_{t}\right)_{t \in \mathcal{T}}$, where the path flow vector $\boldsymbol{y}:=\left(y_{t, s}\right)_{t \in \mathcal{T}, s \in \mathcal{S}_t}$ and the edge load vector $\boldsymbol{f}:=\left(f_e\right)_{e \in \mathcal{E}}$. Then, according to Beckmann \cite{beckmann1956studies}, Wardrop equilibrium can be computed as the solution to the following optimization problem,
\begin{subequations}
    \begin{align}
        \min_{\boldsymbol{y}, \boldsymbol{f}} \ &\sum_{e \in \mathcal{E}} \int_{0}^{f_e}c_e(z)dz\\
        \text{s.t. } \ & (\boldsymbol{y}, \boldsymbol{f}) \in F(\boldsymbol{d}).
    \end{align}
\label{prob:WE_xy}
\end{subequations}
We represent the optimization problem as $W(\boldsymbol{d})$, and the corresponding WE solution pair as $(\widehat{\boldsymbol{y}}, \widehat{\boldsymbol{f}})$. 
\begin{remark}
    When delivering recommendations to user $u$ with the OD pair $\theta_u=t \in \mathcal{T}$, the RS suggests a mixed strategy $p_{u, i}$ over the feasible $s_{u, i} \in \mathcal{S}_u=\mathcal{S}_t$, each $p_{u, i}=y_{t, s_{u, i}}/d_t$.
\end{remark}
Subsequently, consider the situation where a Sybil-based attacker generates non-existent demands $\boldsymbol{d}^a \in \mathbb{Z}_{\ge 0}^{|\mathcal{T}|}$ using Sybil (fake) users. Then, the RS will need to consider an aggregated demand of $\boldsymbol{d}^{\prime}=\boldsymbol{d} + \boldsymbol{d}^a$. Note that for each OD pair $t$, the demand $d_t^{\prime}$ under attack consists of $d_t + d_t^a$. Without loss of generality, we can assume that a proportion of $\frac{d_t}{d_t + d_t^a}$ of the WE expected path flow $\widehat{y}_{t, s}^{\prime}$ with respect to $W(\boldsymbol{d}^{\prime})$ is caused by true users. Hence, denote $\widehat{y}^u_{t, s}=\frac{d_t}{d_t + d_t^a}\widehat{y}_{t, s}^{\prime}$. In a local targeted attack, the attacker aims to generate a certain level of expected flow caused by true app users on the target edge $e^{\prime}$. That is, the attacker aims to make $\sum_{t \in \mathscr{T}}\sum_{s \in \mathscr{S}_t}\widehat{y}^u_{t, s}\mathbf{1}_{\{e^{\prime} \in s\}}$ achieve a desired level $\gamma \in \mathbb{R}_{\ge 0}$.
\begin{subequations}
    \begin{align}
        \min_{\boldsymbol{d}^a} \ &\sum_{t \in \mathcal{T}} d_t^a\\
        \text{s.t.} \ &\sum_{t \in \mathcal{T}}\sum_{s \in \mathcal{S}_t}\widehat{y}^u_{t, s}\mathbf{1}_{\{e^{\prime} \in s\}}\geq \gamma, \\
        &(\widehat{\boldsymbol{y}}, \widehat{\boldsymbol{f}}) \in \argmin W(\boldsymbol{d} + \boldsymbol{d}^a), \\
        & d_t^a \geq 0, \forall t \in \mathcal{T}.
    \end{align}
\label{prob:RS_attack}
\end{subequations}
The attacker's problem can then be computed by Algorithm 2.
\begin{algorithm}
  \caption{\textbf{Mis}informed \textbf{D}emands \textbf{A}ttack}\label{algo:MisDA}
  \begin{algorithmic}[1]
    \State\textbf{Input} $\left\langle \ \mathcal{U}, (\mathcal{S}_u)_{u \in \mathcal{U}}, \mathcal{G}, (c_e(\cdot))_{e \in \mathcal{E}} \ \right\rangle$,
    \State\textbf{Initialize} fabricated demand $\boldsymbol{d^a}$,
      \State\textbf{Obtain} true demand $\boldsymbol{d}$ from $(\mathcal{S}_u)_{u \in \mathcal{U}}\mbox{ and } \mathcal{G}$,
      \While{desired result from $\boldsymbol{d^a}$ is not met}
        \While{$(\boldsymbol{y}, \boldsymbol{f})$ does not reach WE}
        \State gradient descent on \eqref{prob:WE_xy} using PM,
        \EndWhile
        \State\textbf{Obtain} WE $(\widehat{\boldsymbol{y}}, \widehat{\boldsymbol{f}})$,
        \State gradient descent on \eqref{prob:RS_attack} using PM,
      \EndWhile
      \State \textbf{Return} optimal $\boldsymbol{d^a}^*$
  \end{algorithmic}
\end{algorithm}
\vspace{-2mm}
Lastly, we need to notice that the level $\gamma$ in problem \eqref{prob:RS_attack} can not be arbitrarily large, which leads us to the following.
\begin{remark}
    The edge load $f_e$ is bounded by the total demand of the true user $d_T=\sum_{t \in \mathcal{T}}d_t$. Hence, the attacker's desired level $\gamma$ is also upper-bounded by $d_T=\sum_{t \in \mathcal{T}}d_t$.
\end{remark}

\begin{figure*}[!h]
    \centering
    \begin{subfigure}[t]{0.327\textwidth}
        \centering
        \includegraphics[width=\textwidth]{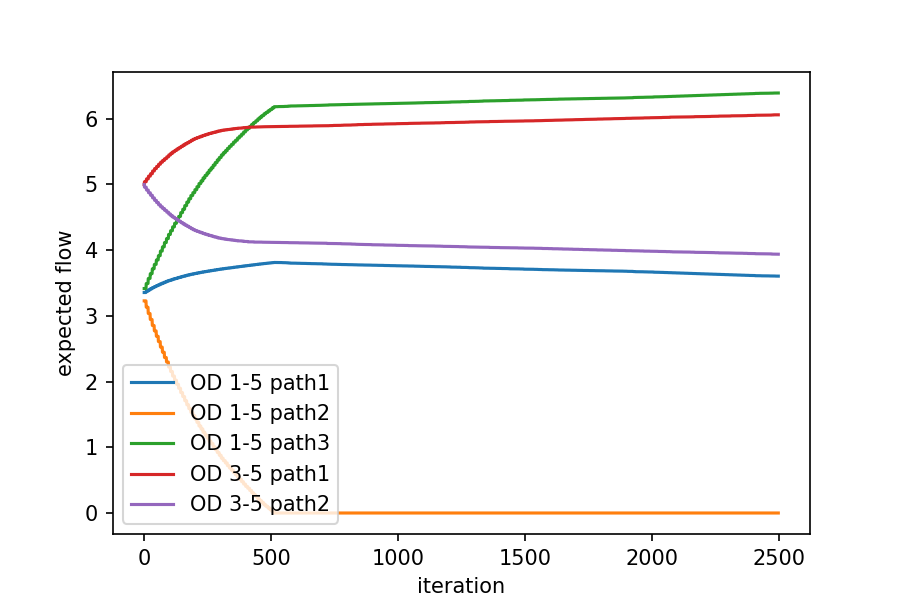}
        \caption{Recommendations without attack.}
        \label{fig:woa_result}
    \end{subfigure}
    \hfill
    \begin{subfigure}[t]{0.327\textwidth}
        \centering
        \includegraphics[width=\textwidth]{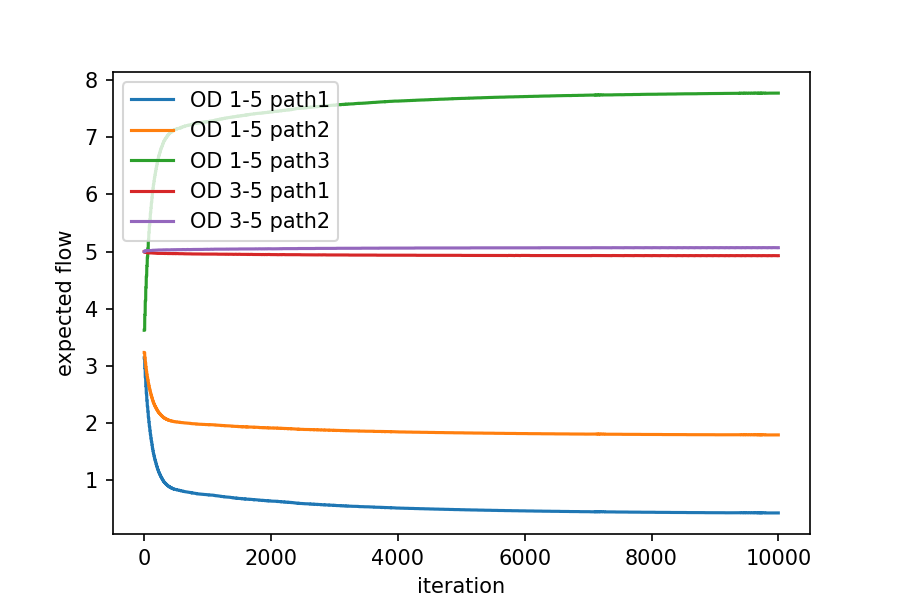}
        \caption{Recommendations under uniform attack.}
        \label{fig:wa_uni_result}
    \end{subfigure}
    \hfill
    \begin{subfigure}[t]{0.327\textwidth}
        \centering
        \includegraphics[width=\textwidth]{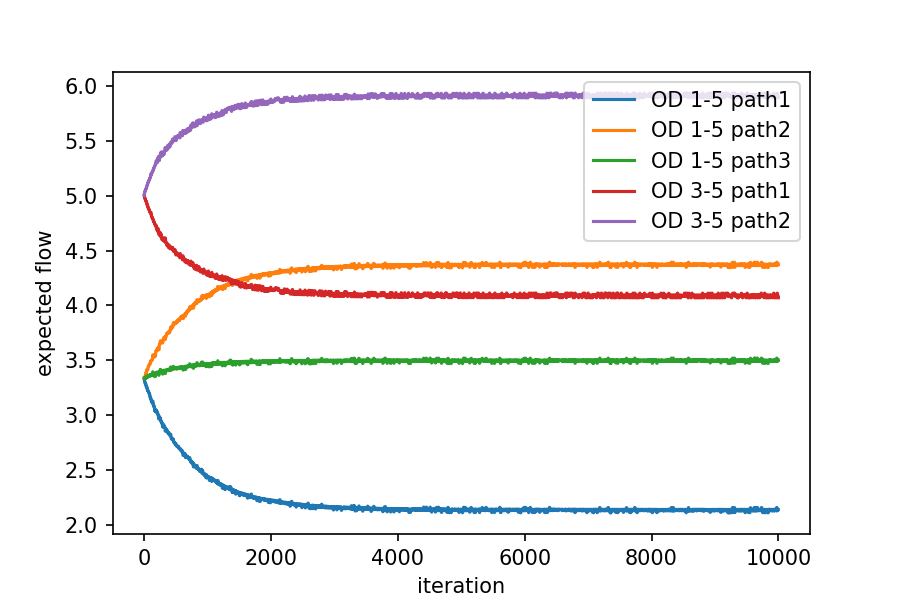}
        \caption{Recommendations under optimal attack.}
        \label{fig:wa_result}
    \end{subfigure}
    \caption{The experiment results. Without attack: OD pairs 1-5 and 3-5 are recommended path 3 and path 1 with a higher probability, respectively. Under uniform attack: OD pairs 1-5 is still recommended path 3 with a higher probability, while both choices for OD pair 3-5 are recommended with a probability close to 0.5. Under optimal attack: OD pairs 1-5 and 3-5 are recommended path 2, containing the targeted road with a higher probability.}
\label{fig:exp_result}
\vspace{-5mm}
\end{figure*}

\section{Characterization for the Misinformed\\ Demands Attack}
\subsection{Optimality Conditions for WE-based Recommendation}
Under the assumption that the cost functions $c_e(f_e)$ are continuous and non-decreasing in $f_e$, a pair $(\boldsymbol{y}, \boldsymbol{f}) \in F(\boldsymbol{d})$ is a minimizer of
$W(\boldsymbol{d})$ if and only if it satisfies the following KKT conditions.
\begin{subequations}
    \begin{align}
        c_e(f_e)-\lambda_e =0, \ &\forall e \in \mathcal{E},\\
        -\nu_t + \sum_{e \in \mathcal{E}}\lambda_e \mathbf{1}_{\{e \in s\}} - \mu_{t, s}=0, \ &\forall s \in \mathcal{S}_t, \forall t \in \mathcal{T},\\
        \mu_{t, s}y_{t, s}=0,\ &\forall s \in \mathcal{S}_t, \forall t \in \mathcal{T},
    \end{align}
\label{eq:KKT}
\end{subequations} 
with Lagrangian multipliers $\nu_t \in \mathbb{R}_{+}, \forall t \in \mathcal{T}$, $\lambda_e \in \mathbb{R}, \forall e \in \mathcal{E}$, and $\mu_{t, s} \geq 0, \forall s \in \mathcal{S}_t, \forall t \in \mathcal{T}$. Then, a pair $(\boldsymbol{y}, \boldsymbol{f})$ satisfying the constraints with multipliers $-\boldsymbol{\nu}=-(\nu_t)_{t\in \mathcal{T}}, \boldsymbol{\lambda}=(\lambda_e)_{e \in \mathcal{E}}, \boldsymbol{\mu}=(\mu_{t, s})_{s \in \mathcal{S}_t, t\in \mathcal{T}}$ also satisfies $$\nu_t=\sum_{e \in s}c_e(f_e)-\mu_{t, s}
\begin{cases}
    =\sum_{e \in s}c_e(f_e), \ y_{t, s} > 0, \\
    \leq \sum_{e \in s}c_e(f_e), \ y_{t, s} = 0, 
\end{cases}$$ which coincides with the definition of WE.
\vspace{-1mm}
\subsection{Impact of Demand Changes on WE}
Then, we aim to examine how the WE pair $(\widehat{\boldsymbol{y}}, \widehat{\boldsymbol{f}})$ can be influenced by changes in the demand $\boldsymbol{d}$ according to \cite{still2018lectures}.

\vspace{-1mm}
\begin{proposition}
    Let the pair $(\boldsymbol{y}, \boldsymbol{f})$ with the corresponding multipliers $\boldsymbol{\nu}$ and $\boldsymbol{\mu}$ described in \eqref{eq:KKT} be a WE for demand $\boldsymbol{d}$ and $(\boldsymbol{y}', \boldsymbol{f}')$ with corresponding multipliers $\boldsymbol{\nu}'$ and $\boldsymbol{\mu}'$ be a WE for demand $\boldsymbol{d}'$. Then, $(\boldsymbol{\nu}'-\boldsymbol{\nu})^{T}(\boldsymbol{d}'-\boldsymbol{d}) \geq \boldsymbol{\mu}'^{T}\boldsymbol{y}+\boldsymbol{\mu}^{T}\boldsymbol{y}' \geq 0$.
\label{prop:diff_d}
\end{proposition}
The result of the above proposition can also be written as
$$\left[\sum_{e \in s}c_e(f_e')-c_e(f_e)\right](d_t'-d_t)\geq 0,$$ for $s \in \mathcal{S}_t$ with $y_{t, s}', y_{t, s}>0$.
The Proposition \ref{prop:diff_d} states that if one demand $d_t$ is increased by Sybil users, with other demands remaining the same, then the equilibrium cost $\nu_t$ calculated by the RS for the user $u$ with OD pair $\theta_u=t$ is also increased. 
\begin{proposition}
    For $W(\boldsymbol{d})$ with demand $\boldsymbol{d}$, let $\boldsymbol{f}$ be a WE corresponds to cost $c_e(f_e)$ and $\boldsymbol{f}'$ be a WE corresponds to cost $c_e'(f_e)$, then $\left[c_e'(f_e)-c_e(f_e)\right](f_e'-f_e) \leq 0$ and $\left[c_e'(f_e')-c_e(f_e')\right](f_e'-f_e) \leq 0$.
\end{proposition}
Then, we show that an increasing cost on an edge $e \in \mathcal{E}$ will cause the equilibrium load on that edge $f_e$ to decrease. The reduced load can be interpreted as a redistribution to alternative feasible paths. Hence, the attacker can fulfill the goal if the load is redistributed to the desired road $e^{\prime}$.

\section{Numerical Experiments}\label{sec:experiment}

We use the traffic network described in Fig. \ref{fig:example_2} as a case study of the RS and the misinformed demand attack. Both RS and the attacker's problem can be addressed using gradient descent with PM. In this case, we simply consider users with two OD pairs: one (blue) aims to go from node $1$ to $5$, and the other (green) wants to go from node $3$ to $5$. The set of feasible paths for OD 1-5 is $\mathcal{S}_1=\{1, 2, 3\}$ with choice $1$ being path 1-3-4-5, choice $2$ being path 1-3-5, and choice $3$ being path 1-2-5. The set of feasible paths for OD 3-5 is $\mathcal{S}_1=\{1, 2\}$ with choice $1$ being path 3-4-5, and choice $2$ being path 3-5. The number displayed on each edge represents the free-flow time cost $t_e$.

\subsection{Recommendation System}
For simplicity, the cost function $c_e(\cdot)$ is selected as the travel time cost $c_e(f_e)=t_e(1+0.4\left(\frac{f_e}{k_e} \right)^2)$, where $t_e$ is the free-flow travel time on edge $e$ indicated red in Fig. \ref{fig:example_2}, and the capacity $k_e, \forall e \in \mathcal{E}$ is chosen as $10$. In the case of independent routing, users with OD pair 1-5 will either choose choice 1 or 3, and users with OD pair 3-5 will choose choice 1. In this case, the worst case (all users decide on choice 1) total travel time cost will be $236$ seconds ($10.4$ sec for path 3-4-5 and $13.2$ sec for path 1-3-4-5). However, if users follow the recommendation in Fig. \ref{fig:woa_result}, the total travel time cost will be $126$ seconds with both choices for OD pair 3-5 being $5.3$ sec and both choice 1 and 3 for OD pair 1-5 being $7.3$ sec. As a result, the proposed RS can help achieve a better outcome (lower total cost) while not sacrificing users' performance (in terms of expected travel time cost). 

\begin{figure}[h]
    \centering
    \includegraphics[width=3.35in]{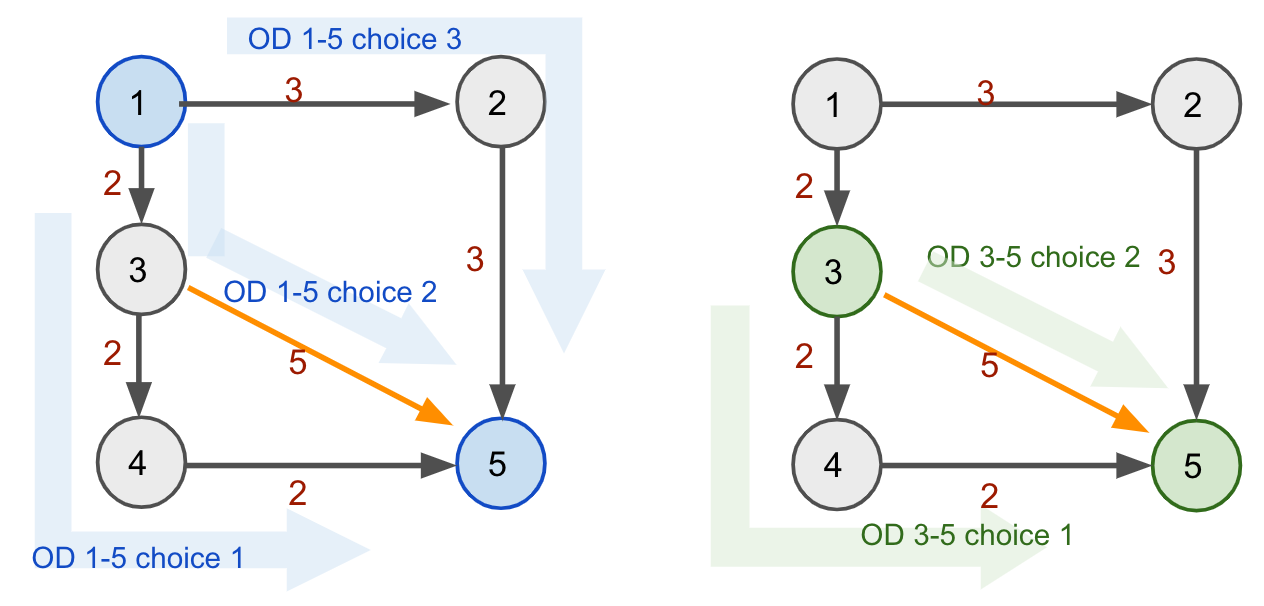}
    \caption{The setting where two OD pairs (blue and green) are presented with multiple viable paths. The attacker's goal is to manipulate the recommendations to the users, coaxing them to go through paths containing the orange edge. That is, the attacker desires choice 2 of both OD pairs. The free-flow road costs are indicated in red on the edges.}
    \label{fig:example_2}
    \vspace{-6mm}
\end{figure}

\subsection{Under Attacks}

We consider the case where the attacker can fabricate user demands for a set $\mathcal{K} \subset \mathcal{V} \times \mathcal{V}$ of distinct OD pairs to mislead the RS who selects WE as a feasible recommendation. We have $10$ authentic users aiming to go from node $1$ to $5$ (10 demands for OD pair 1-5), and $10$ users want to go from node $3$ to $5$ (10 demands for OD pair 3-5), the attacker can identify the desired fake demand levels by solving the problem in section \ref{sec:md_problem}. As shown in Fig. \ref{fig:wa_result}, we can observe that both OD pairs 1-5 and 3-5 exhibit a higher probability of passing the target edge $(3, 5)$. The expected flow caused by authentic users illustrated in Fig. \ref{fig:edge_flow} also shows that $f^r_{(3, 5)}$ meets the desired level $\gamma \geq 10$ by generating a total of $30$ non-existent demands on the traffic network.

For benchmarking purposes, we select two relatively straightforward attacker profiles: uniform and random. These attackers are restricted to allocating the same amount of demands as the optimal attacker, totaling $30$ demands. The uniform attacker evenly distributes the total demand across all OD pairs, while the random attacker distributes the demand randomly among the OD pairs. It should be noted that, in the case of the random attacker, we are interested in presenting its average performance of $200$ experiments.
\vspace{-3mm}
\begin{figure}[H]
    \centering
    \includegraphics[width=2.8in]{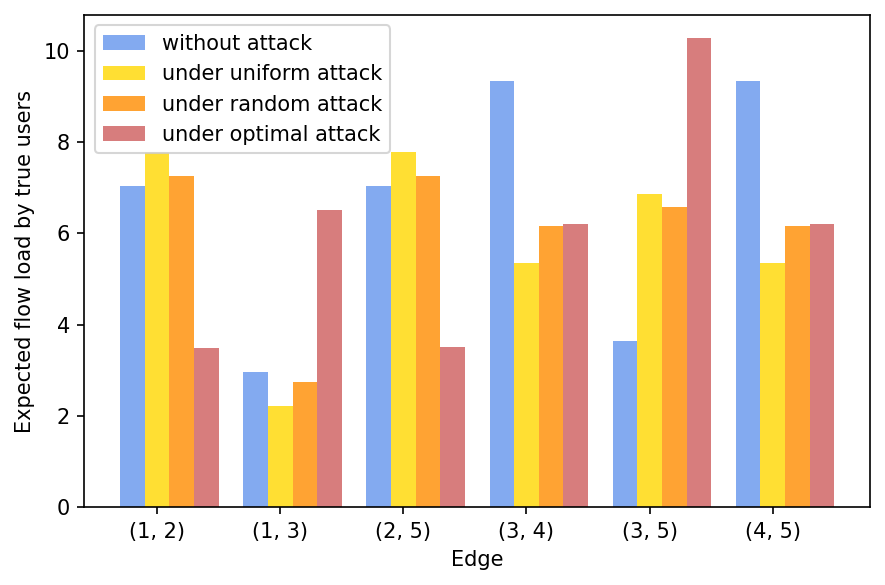}
    \caption{The expected edge flow caused by authentic users with and without attack. In this case, the targeted edge is $(3, 5)$ and the desired flow level is $\gamma=10$.}
    \label{fig:edge_flow}
    \vspace{-4mm}
\end{figure}

As shown in Fig. \ref{fig:edge_flow}, we compare the expected flow load on each edge generated by authentic users. In this comparison, all attacker profiles contribute to an increase in the expected flow on our targeted edge $(3,5)$. However, it is essential to note that the optimal attacker yields the most significant increase in flow on this edge and is the only one that reaches the desired level $\gamma=10$.

\subsection{From Braess' Paradox to Resilience Paradox}

To this end, a natural question is: \textit{Can the local targeted attack lead to a better overall outcome (total costs for drivers and users) in some situations?} We can start with a carefully crafted example leveraging the classical Braess’ network \cite{Braess}. 

\subsubsection{Misinformed Demands} \label{sec:paradox_mis_demands}
Within the transportation network shown in Fig. \ref{fig:paradox}, there are $30$ users aiming to go from node A to node B, and the $\epsilon$ is small enough so that the cost on C-D is close to $0$ even though all $30$ users are passing through. Before the attack (illustrated in Fig. \ref{fig:paradox1}), the RS will recommend a mixed strategy $(1/3, 1/3, 1/3)$ on path A-C-B, A-C-D-B, and A-D-B, respectively. The overall costs on these three paths are all $4$, which leads to a total travel time cost of $120$ for users. Suppose the attacker wants more ``users'' to pass D-B by fabricating a large demand on C-D to make C-D seem congested to the RS, as in Fig. \ref{fig:paradox2}. The RS will recommend a strategy $(0.5, 0, 0.5)$ on paths A-C-B, A-C-D-B, and A-D-B, respectively. The overall costs on A-C-B and A-D-B are both $3.5$, which leads to the total travel time cost for users becoming $105$. The cost under attack is better than the performance without attack.
\vspace{-3mm}
\begin{figure}[!h]
    \centering
    \begin{subfigure}[t]{0.21\textwidth}
        \centering
        \includegraphics[width=\textwidth]{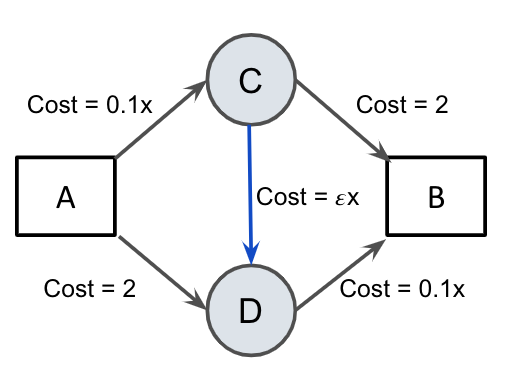}
        \caption{Without attack}
        \label{fig:paradox1}
    \end{subfigure}
    \hfill
    \begin{subfigure}[t]{0.21\textwidth}
        \centering
        \includegraphics[width=\textwidth]{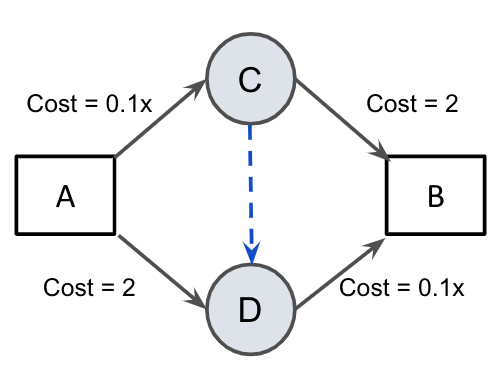}
        \caption{Under attack}
        \label{fig:paradox2}
    \end{subfigure}
    \caption{A carefully crafted example for the discussion on the ``Resilience Paradox''.}
\label{fig:paradox}
\vspace{-3mm}
\end{figure}

\subsubsection{Misinformed Costs}
Consider an alternative attack scenario where the attacker can introduce misinformed traffic conditions by manipulating the cost functions associated with edges, rather than fabricating non-existent demands. For illustration purposes, using the same example as depicted in Fig. \ref{fig:paradox}, where there are $30$ users aiming to travel from A to B, the recommendation without attack is $(1/3, 1/3, 1/3)$, resulting in a total cost of $120$. In the case of a misinformed cost attack, the attacker can mislead the RS into a similar recommendation, specifically $(0.5, 0, 0.5)$, along with a reduced total cost of $105$, as discussed in Section \ref{sec:paradox_mis_demands}, by misinforming the coefficient $\epsilon$ associated with the cost of C-D to a sufficiently high value.

\subsubsection{Misinformation in Changing Other Drivers' Behavior}
Suppose that there are $20$ drivers and $10$ users aim to go from node A to node B, and the transportation network is shown in Fig. \ref{fig:paradox}. Before the attack, the path A-C-D-B seems better for other less strategic drivers. Considering other drivers' behavior, the RS will recommend a mixed strategy $(0.5, 0, 0.5)$ on paths A-C-B, A-C-D-B, and A-D-B, respectively. The overall costs on these three paths are $4.5, 5$, and $4.5$, which leads to the total travel time cost for drivers and users becoming $145$. If the attacker's objective is to divert more ``users'' to route D-B by simulating a car accident on road C-D to misguide drivers, it results in drivers' behavioral response being $(0.5, 0, 0.5)$. In this case, the RS will recommend a mixed strategy $(0, 1, 0)$ on paths A-C-B, A-C-D-B, and A-D-B, respectively. The overall costs on these three paths are all $4$, resulting in an overall travel time cost of 120 for both drivers and users. This outcome represents an improvement over the performance in the absence of the attack. We call this phenomenon ``Resilience Paradox'',  where attacks help improve the performance of the recommendations, and strengthen the resilience of the network. 

We summarize the scenarios in the carefully constructed example involving a total demand of $30$ drivers and users discussed above as Table \ref{tab:TTC}. The relation of total travel time costs between these cases not only illustrates the ``Resilience Paradox'' but also emphasizes the significance of the RS. The latter is evident that a lower total cost is attained when the demand composition shifts from $20$ drivers and $10$ users to $30$ users in scenarios without attacks.

\begin{table}[ht]
    \centering
    \begin{tabular}{|c|c|c|c|}
         \hline
         \textbf{Drivers} & \textbf{Users} & \textbf{Attack} & \textbf{Travel Time Costs} \\
         \hline
         20 & 10 & No & 145 \\
         20 & 10 & Yes & 120 \\
         0 & 30 & No & 120 \\
         0 & 30 & Yes & 105 \\
         \hline
    \end{tabular}
    \caption{Travel Time Costs in Different Scenarios}
    \label{tab:TTC}
\vspace{-5mm}
\end{table}

\section{Conclusions}
This paper aims to explore potential informational attacks on navigational RSs. We first propose an RS that considers human factors such as users' compliance and other non-user drivers' behaviors when providing recommendations. Then, we identify various avenues attackers can exploit to benefit certain groups or elevate traffic congestion levels and focus on the misinformed demand attack in a local targeted sense. The attacker’s problem is analyzed based on the Stackelberg game framework rooted in Wardrop equilibrium with additional non-existent demands. 

Our study shows that users are vulnerable to attacks targeting specific roads by creating non-existent demands with OD pairs associated with alternative path options. Through the exploration, we highlight the need to consider informational attacks for a more resilient and effective navigational recommendation in the future. In addition, one of the other possible future directions can be investigating different scenarios, such as the impact of misinformed traffic conditions (costs) attacks on the RS. In contrast to the network-wide attacks that worsen traffic congestion, we also illustrate the ``Resilience Paradox'', in which the locally targeted attack by misinformed demands or misinformed traffic conditions can benefit the overall traffic outcome in terms of total travel time costs. This also points out that the locally targeted attack is a potential aspect worth further investigation.

\bibliographystyle{IEEEtran}
\bibliography{IEEEabrv,reference}

\begin{thebibliography}{10}
\providecommand{\url}[1]{#1}
\csname url@samestyle\endcsname
\providecommand{\newblock}{\relax}
\providecommand{\bibinfo}[2]{#2}
\providecommand{\BIBentrySTDinterwordspacing}{\spaceskip=0pt\relax}
\providecommand{\BIBentryALTinterwordstretchfactor}{4}
\providecommand{\BIBentryALTinterwordspacing}{\spaceskip=\fontdimen2\font plus
\BIBentryALTinterwordstretchfactor\fontdimen3\font minus \fontdimen4\font\relax}
\providecommand{\BIBforeignlanguage}[2]{{%
\expandafter\ifx\csname l@#1\endcsname\relax
\typeout{** WARNING: IEEEtran.bst: No hyphenation pattern has been}%
\typeout{** loaded for the language `#1'. Using the pattern for}%
\typeout{** the default language instead.}%
\else
\language=\csname l@#1\endcsname
\fi
#2}}
\providecommand{\BIBdecl}{\relax}
\BIBdecl

\bibitem{lv2020ai}
Z.~Lv, R.~Lou, and A.~K. Singh, ``Ai empowered communication systems for intelligent transportation systems,'' \emph{IEEE Transactions on Intelligent Transportation Systems}, vol.~22, no.~7, pp. 4579--4587, 2020.

\bibitem{zantalis2019review}
F.~Zantalis, G.~Koulouras, S.~Karabetsos, and D.~Kandris, ``A review of machine learning and iot in smart transportation,'' \emph{Future Internet}, vol.~11, no.~4, p.~94, 2019.

\bibitem{veres2019deep}
M.~Veres and M.~Moussa, ``Deep learning for intelligent transportation systems: A survey of emerging trends,'' \emph{IEEE Transactions on Intelligent transportation systems}, vol.~21, no.~8, pp. 3152--3168, 2019.

\bibitem{haydari2020deep}
A.~Haydari and Y.~Y{\i}lmaz, ``Deep reinforcement learning for intelligent transportation systems: A survey,'' \emph{IEEE Transactions on Intelligent Transportation Systems}, vol.~23, no.~1, pp. 11--32, 2020.

\bibitem{van2016user}
M.~van Essen, T.~Thomas, E.~van Berkum, and C.~Chorus, ``From user equilibrium to system optimum: a literature review on the role of travel information, bounded rationality and non-selfish behaviour at the network and individual levels,'' \emph{Transport reviews}, vol.~36, no.~4, pp. 527--548, 2016.

\bibitem{8262884}
S.~Das, E.~Kamenica, and R.~Mirka, ``Reducing congestion through information design,'' in \emph{2017 55th Annual Allerton Conference on Communication, Control, and Computing (Allerton)}, 2017, pp. 1279--1284.

\bibitem{ning2023robust}
Y.~Ning and L.~Du, ``Robust and resilient equilibrium routing mechanism for traffic congestion mitigation built upon correlated equilibrium and distributed optimization,'' \emph{Transportation research part B: methodological}, vol. 168, pp. 170--205, 2023.

\bibitem{mecheva2020cybersecurity}
T.~Mecheva and N.~Kakanakov, ``Cybersecurity in intelligent transportation systems,'' \emph{Computers}, vol.~9, no.~4, p.~83, 2020.

\bibitem{Waniek2021}
M.~Waniek, G.~Raman, B.~AlShebli, J.~C.-H. Peng, and T.~Rahwan, ``Traffic networks are vulnerable to disinformation attacks,'' \emph{Scientific Reports}, vol.~11, no.~1, p. 5329, March 2021.

\bibitem{waze_police}
Sophos, ``Are miami cops really flooding waze with fake police sightings,'' https://nakedsecurity.sophos.com/2015/02/16/are-miami-cops-really-flooding-waze-with-fake-police-sightings/.

\bibitem{waze_resident}
S.~News, ``Waze to go: Residents fight off crowdsourced traffic for a while,'' https://nakedsecurity.sophos.com/2016/06/07/waze-to-go-residents-fight-off-crowdsourced-traffic-for-a-while/.

\bibitem{eryonucu2022sybil}
C.~Eryonucu and P.~Papadimitratos, ``Sybil-based attacks on google maps or how to forge the image of city life,'' in \emph{Proceedings of the 15th ACM Conference on Security and Privacy in Wireless and Mobile Networks}, 2022, pp. 73--84.

\bibitem{Multihoming}
\BIBentryALTinterwordspacing
D.~K. Goldenberg, L.~Qiuy, H.~Xie, Y.~R. Yang, and Y.~Zhang, ``Optimizing cost and performance for multihoming,'' in \emph{Proceedings of the 2004 Conference on Applications, Technologies, Architectures, and Protocols for Computer Communications}, ser. SIGCOMM '04.\hskip 1em plus 0.5em minus 0.4em\relax New York, NY, USA: Association for Computing Machinery, 2004, p. 79–92. [Online]. Available: \url{https://doi.org/10.1145/1015467.1015478}
\BIBentrySTDinterwordspacing

\bibitem{zhang2010optimizing}
Z.~Zhang, M.~Zhang, A.~G. Greenberg, Y.~C. Hu, R.~Mahajan, and B.~Christian, ``Optimizing cost and performance in online service provider networks.'' in \emph{NSDI}, 2010, pp. 33--48.

\bibitem{WE}
R.~Cominetti, M.~Scarsini, M.~Schr{\"{o}}der, and N.~E.~S. Moses, ``Convergence of large atomic congestion games,'' \emph{CoRR}, vol. abs/2001.02797, 2020.

\bibitem{siuhi2016opportunities}
S.~Siuhi and J.~Mwakalonge, ``Opportunities and challenges of smart mobile applications in transportation,'' \emph{Journal of traffic and transportation engineering (english edition)}, vol.~3, no.~6, pp. 582--592, 2016.

\bibitem{flash_crowd_effect}
A.~Grzybek, G.~Danoy, P.~Bouvry, and M.~Seredynski, ``Mitigating flash crowd effect using connected vehicle technology,'' \emph{Vehicular Communications}, vol.~2, no.~4, pp. 238--250, 2015.

\bibitem{user_eqm_rl}
B.~Zhou, Q.~Song, Z.~Zhao, and T.~Liu, ``A reinforcement learning scheme for the equilibrium of the in-vehicle route choice problem based on congestion game,'' \emph{Applied Mathematics and Computation}, vol. 371, p. 124895, 2020.

\bibitem{cybersecurity_transit}
D.~Fletcher and P.~Bye, ``Cybersecurity in transit systems,'' Tech. Rep., 2022.

\bibitem{pan2022poisoned}
Y.~Pan and Q.~Zhu, ``On poisoned wardrop equilibrium in congestion games,'' in \emph{International Conference on Decision and Game Theory for Security}.\hskip 1em plus 0.5em minus 0.4em\relax Springer, 2022, pp. 191--211.

\bibitem{yen1971finding}
J.~Y. Yen, ``Finding the k shortest loopless paths in a network,'' \emph{management Science}, vol.~17, no.~11, pp. 712--716, 1971.

\bibitem{daskin1985urban}
M.~S. Daskin, ``Urban transportation networks: Equilibrium analysis with mathematical programming methods,'' 1985.

\bibitem{lee2018comparison}
D.~Lee, S.~Derrible, and F.~C. Pereira, ``Comparison of four types of artificial neural network and a multinomial logit model for travel mode choice modeling,'' \emph{Transportation Research Record}, vol. 2672, no.~49, pp. 101--112, 2018.

\bibitem{sorin2015finite}
S.~Sorin and C.~Wan, ``Finite composite games: Equilibria and dynamics,'' \emph{arXiv preprint arXiv:1503.07935}, 2015.

\bibitem{beckmann1956studies}
M.~Beckmann, C.~B. McGuire, and C.~B. Winsten, ``Studies in the economics of transportation,'' Tech. Rep., 1956.

\bibitem{still2018lectures}
G.~Still, ``Lectures on parametric optimization: An introduction,'' \emph{Optimization Online}, p.~2, 2018.

\bibitem{Braess}
X.~Di, X.~He, X.~Guo, and H.~X. Liu, ``Braess paradox under the boundedly rational user equilibria,'' \emph{Transportation Research Part B: Methodological}, vol.~67, pp. 86--108, 2014.

\end{thebibliography}
\vspace{12pt}

\end{document}